\documentclass[10pt,a4paper]{article}
\usepackage[utf8]{inputenc}
\usepackage[english]{babel}
\usepackage{amsmath}
\usepackage{amsfonts}
\usepackage{amssymb}
\usepackage{amsthm}
\usepackage{cite}
\usepackage{graphicx}
\usepackage[left=2cm,right=2cm,top=2cm,bottom=2cm]{geometry}
\usepackage{hyperref}
\hypersetup{
    colorlinks=true,
    linkcolor=blue,
    filecolor=magenta,      
    urlcolor=cyan,
    pdfpagemode=FullScreen,
    }
\newtheorem{te}{Theorem}

\newtheorem{de}{Definition}
\newtheorem{ex}{Example}
\providecommand{\keywords}[1]
{\small	\textbf{\textit{Keywords:}} #1}
\author{R. Azuaje and A. M. Escobar-Ruiz\\
Departamento de F\'{i}sica, Universidad Aut\'onoma Metropolitana Unidad Iztapalapa,\\ San Rafael Atlixco 186, 09340, Ciudad de M\'exico, M\'exico}
\title{Canonical transformations: from the coordinate based approach to the geometric one}
\begin{document}
\maketitle

\begin{abstract}

In this paper the theory of time-dependent and time-independent canonical transformations is considered from a geometric perspective. Both the
geometric formalism and the coordinate based approach are described in detail. In particular, one-parameter groups of canonical transformations are geometrically identified with flows of Hamiltonian vector fields which, in turn, are their infinitesimal generators. Likewise, infinitesimal generators of invariance transformations are geometrically characterized. The main results are established in the form of theorems and the connection between the geometric and the coordinate based frameworks is remarked using concrete examples. 


\end{abstract}

\section{Introduction}

In the theory of Classical Hamiltonian Mechanics, the canonical transformations play a fundamental role \cite{landau1982mechanics, GPS2002,Calkin96} due to many reasons. For instance, they preserve the form of the Hamilton's equations of motion (invariance), deprives the generalized coordinates $q$ and canonical momenta $p$ of a significant
part of their original meaning, and more importantly, they can be used to simplify the task of solving the equations of motion. In particular, they are instrumental to exploit the existence of conserved quantities, and represent the underlying central object within the formalism of Hamilton-Jacobi theory. Interestingly, the time evolution of the dynamical canonical variables $q(t), p(t)$ during the
motion can be viewed itself as a series of canonical transformations \cite{landau1982mechanics}. Frequently, in practice, the canonical transformations are defined and used in what we can call the coordinate based approach. 

The elegant free-coordinate geometric approach to Classical Mechanics has a long history as well, starting with the pioneering works of Arnold \cite{Arnold78}, and Abraham and Marsden \cite{AMRC2008}. Nowadays it is an intense active research area. Needless to say that the natural geometric framework to develop the theory of Classical time-independent Hamiltonian Mechanics is the symplectic geometry \cite{AMRC2008,LR89,Torres2020,Lee2012,marsden2013introduction,AKN2006}. Under this formalism, the canonical transformations have been introduced as transformations preserving the symplectic structure of the corresponding phase space \cite{AMRC2008,Lee2012}. A more general definition of canonical transformation, including those preserving the symplectic structure up to a conformal constant factor, is given in \cite{CFR2013}. However, to the best of the knowledge of the present authors, a geometric description of time-dependent canonical transformations has not been addressed in the literature. 

The main goal of this paper is to present a unified \textit{coordinate-geometric} description of canonical transformations. From one side, the coordinate approach of canonical transformations is widely known even for undergraduate students of any physics program. On the other hand, the geometric theory of canonical transformations has been much less studied and, as we will see, some important aspects that have not been analyzed before occur. It is surprising that in the literature a clear explanation on the connection between the coordinate and geometric approaches is missing. Specifically, in spite of the fact that the relation between invariance and conserved quantities is well understood (in both the Lagrangian and the Hamiltonian formulations), and the relation of conserved quantities with symmetries is established by the celebrated Noether theorem, a natural question still arises \textit{What is the connection between the notion of invariance (coordinate approach) and the symmetries described by vector fields (geometric approach)?}. Therefore, to answer this question we aim to elaborate the aforementioned coordinate-geometric description of canonical transformations which, in turn,  allows us to geometrically visualize one-parameter groups of canonical transformations as flows of Hamiltonian vector fields. This visualization embodies one of the main results of the present consideration. In addition to the rigorous mathematical treatment, practical aspects of this study  are highlighted as well.

This paper is organized as follows. In section \ref{seccanonicalclassical}, a brief review of the classical description of canonical transformations regarded as coordinate transformations is presented. The concepts of infinitesimal transformations and invariance transformations are also reviewed. In section \ref{seccanonicalgeo}, we cover the theory of canonical transformation from a pure geometrical point of view. This section is divided into five subsections. The geometric formalisms of time-independent and time-dependent Hamiltonian Mechanics are briefly resumed in the first one. Next, the canonical transformations are presented for time-independent and time-dependable Hamiltonian systems in the second and third subsections, respectively. Especially, it is shown in a clear and systematic way how the modern geometric definition of canonical transformation coincides locally with the classical coordinate description. In the subsection \ref{subseconeparameter}, the concept of one-parametric group of canonical transformations is revisited within the geometric formalism. Afterwards, in section \ref{subsecsymmetries}, the notion of invariance is introduced from a pure geometric perspective and it is related to the symmetries described by vector fields. This represents one of the most important result of our modest study. Throughout the paper, the novel ideas, key concepts and developments are indicated and explained in detail. Finally, in section \ref{secexample}, further additional illustrating examples are presented.

\section{Canonical transformations: the coordinate-based viewpoint}
\label{seccanonicalclassical}

Let us consider a Hamiltonian system with $n$ degrees of freedom whose equations of motion are given in a set of canonical coordinates $(q,p)=(q^{1},\ldots,q^{n},p_{1},\ldots,p_{n})$ on phase space with a possibly time-dependent Hamiltonian function $H=H(q,p,t)$, i.e., the equations of motion are given by
\begin{equation}
\left\lbrace \begin{array}{c}
\dot{q}^{1} = \frac{\partial H}{\partial p_{1}},\\ 
\vdots \\
\dot{q}^{n} = \frac{\partial H}{\partial p_{n}},\\ \\
\dot{p}_{1} = -\frac{\partial H}{\partial q^{1}},\\
\vdots\\
\dot{p}_{n} = -\frac{\partial H}{\partial q^{n}} \ .
\end{array} \right.
\end{equation}
The Poisson bracket of two possibly time-dependent functions $f=f(q,p,t)$ and $g=g(q,p,t)$ on the canonical coordinates $(q,p)$ take the form
\begin{equation}
\lbrace f,g\rbrace_{q,p} \ = \  \frac{\partial f}{\partial q^{i}}\frac{\partial g}{\partial p_{i}}\,-\,\frac{\partial f}{\partial p_{i}}\frac{\partial g}{\partial q^{i}} \ .
\end{equation}
In this section it is presented a brief review of the theory of canonical transformations from the coordinate approach as described in classical textbooks on Classical Mechanics like \cite{Calkin96,GPS2002,landau1982mechanics}

\subsection{Canonical transformations and generating functions}
\begin{de}
A coordinate transformation $(q,p)\mapsto (Q(q,p,t),P(q,p,t))$ is a canonical transformation if the new set of coordinates $(Q,P)$ forms a set of canonical coordinates.
\end{de}
It is well known that a set of coordinates $(Q^{1},\ldots,Q^{n},P_{1},\ldots,P_{n})$ forms a set of canonical coordinates if and only if
\begin{equation}
\lbrace Q^{i},P_{j}\rbrace_{q,p} = \delta^{i}_{j} \quad\text{and}\quad \lbrace Q^{i},Q_{j}\rbrace_{q,p}=\lbrace P^{i},P_{j}\rbrace_{q,p}=0 \ , \quad i,j=1,2,\ldots,n \ ,
\end{equation}
where $\delta^{i}_{j}$ is the Kronecker delta, which implies that canonical transformations are independent of the Hamiltonian function. Canonical transformations preserve the form of the Hamilton's equations of motion, i.e., if $(q,p)\mapsto (Q(q,p,t),P(q,p,t))$ is a canonical transformation then there exists a new Hamiltonian function $K=K(Q,P,t)$ such that 
\begin{equation}
\left\lbrace \begin{array}{c}
\dot{Q}^{1} = \frac{\partial K}{\partial P_{1}},\\ 
\vdots \\
\dot{Q}^{n} = \frac{\partial K}{\partial P_{n}},\\ \\
\dot{P}_{1} = -\frac{\partial K}{\partial Q^{1}},\\
\vdots\\
\dot{P}_{n} = -\frac{\partial K}{\partial Q^{n}} \ .
\end{array} \right.
\end{equation}

Given a canonical transformation $(q,p)\mapsto (Q,P)$ there exists a possibly time-dependent function $F$, defined on phase space, which obeys
\begin{equation}
p_{i}dq^{i}-Hdt-(P_{i}dQ^{i}-Kdt)\ = \ dF \ .
\label{Fdef}
\end{equation}
$F$ is called the generating function of the canonical transformation. Hence, in this case we have a set of $4n$ coordinates $(q^{1},\ldots,q^{n},p_{1},\ldots,p_{n},Q^{1},\ldots,Q^{n},P_{1},\ldots,P_{n})$ in phase space; $2n$ coordinates are independent and the other $2n$ are dependent. Accordingly, it is convenient to consider canonical transformations consisting of $n$ old coordinates plus $n$ new coordinates all of them functionally independent. There are four basic types of canonical transformations (although mixtures of these can appear) described by sets of the form $(q,Q)$, $(q,P)$, $(p,Q)$ or $(p,P)$, respectively. If the coordinates $(q,Q)$ are independent, the function $F$ can be written as $F=F_{1}(q,Q,t)$ satisfying the condition $\det(\frac{\partial^{2}F_{1}}{\partial q^{i}\partial Q^{j}})\neq 0$, and from (\ref{Fdef}) we also have the following relations
\begin{equation}
p_{i}=\frac{\partial F_{1}}{\partial q^{i}}, \quad P_{i}=-\frac{\partial F_{1}}{\partial Q^{i}}\quad\text{and}\quad K-H=\frac{\partial F_{1}}{\partial t}.
\end{equation}
In this case it is said that $F=F_{1}(q,Q,t)$ is a generating function of type 1. For the identity transformation $(q=Q,p=P)$, which will play an important role in the next Sections, the coordinates $(q,Q)$ are dependent. Thus, in this particular case we can choose the variables $(q,P)$ as independent coordinates and construct, by means of the appropriate Legendre transformation, the generating function $F_{2}(q,P)=F_{1}+Q^{i}P_{i}$. From (\ref{Fdef}) we obtain
\begin{equation}
p_{i}dq^{i}-Hdt+Q^{i}dP_{i}+Kdt=dF_{2} \ , 
\end{equation}
then
\begin{equation}
p_{i}=\frac{\partial F_{2}}{\partial q^{i}}, \quad Q^{i}=\frac{\partial F_{2}}{\partial P_{i}}\quad\text{and}\quad K-H=\frac{\partial F_{2}}{\partial t}.
\end{equation}
In this case it is said that $F_{2}(q,P,t)$ is a generating function of type 2. In particular, the generating function 
\begin{equation}
F_{Id}(q,P) \ \equiv \ q^{i}\,P_{i} \ ,  
\end{equation}
yields the identity canonical transformation.

\subsection{Infinitesimal canonical transformations}
\begin{de}
An infinitesimal canonical transformation is a transformation ¨near¨ the identity, i.e., a coordinate transformation of the form
\begin{equation}
\label{eqinfcan}
Q^{i}=q^{i}+\epsilon\, \Delta q^{i}(q,p,t), \quad P_{i}=p_{i}+\epsilon \,\Delta p_{i}(q,p,t)\ ,
\end{equation}
where $\epsilon$ is an infinitesimal parameter (the change from $(q,p)$ to $(Q,P)$ is ¨small¨) and, $\Delta q^{i}(q,p,t)$ and $\Delta p_{i}(q,p,t)$ are certain functions to be determined.
\end{de}
Infinitesimal canonical transformations are generated by functions ¨near¨ the generating function $F_{Id}(q,P)$ of the identity transformation. They can be written as follows
\begin{equation}
\label{F2iCT}
F_{2}(q,P,t)\ = \ q^{i}P_{i} \ + \ \epsilon\, G(q,P,t)\ ,
\end{equation}
here $G(q,P,t)$ is any function. Indeed, the canonical transformation generated (\ref{F2iCT}) reads
\begin{equation}
Q^{i}\ = \ q^{i}\ + \ \epsilon\,\frac{\partial G(q,P,t)}{\partial P_{i}} \quad{\rm and}\quad P_{i}\ = \ p_{i}\ - \ \epsilon\, \frac{\partial G(q,P,t)}{\partial q^{i}}\ .
\end{equation}
Evidently, at order $\epsilon$, one can substitute $P_{i}=p_{i}$ into $\epsilon\,\frac{\partial G(q,P,t)}{\partial q^{i}}$ and $\epsilon\,\frac{\partial G(q,P,t)}{\partial P_{i}}$. Hence, we obtain the infinitesimal canonical transformation explicitly 
\begin{equation}
\label{eqgeninfcan}
Q^{i}\ = \ q^{i}\ + \ \epsilon\,\frac{\partial G(q,p,t)}{\partial p_{i}} \quad{\rm and}\quad P_{i}\ = \ p_{i}\ - \ \epsilon\, \frac{\partial G(q,p,t)}{\partial q^{i}}\ .
\end{equation}
From (\ref{eqinfcan}) and (\ref{eqgeninfcan}),  we can make the identifications $\Delta q^{i}(q,p,t)=\frac{\partial G(q,p,t)}{\partial p_{i}}$ and $\Delta p_{i}=-\frac{\partial G(q,p,t)}{\partial q_{i}}$. In this case, the function $G=G(q,p,t)$ in (\ref{eqgeninfcan}) is called the generating function of the  infinitesimal canonical transformation.

\subsection{Invariance transformations and constants of motion}
\label{subsec2.3}
Given a canonical transformation $(q,p)\mapsto(Q(q,p,t),P(q,p,t))$ with generating function $F$ (of any type), the new Hamiltonian function $K$ is given by
\begin{equation}
K(q,p,t)=H(q,p,t)+\frac{\partial F}{\partial t}.
\end{equation}
For a time-independent canonical transformation the new Hamiltonian can be taken equal to the old one.
 
It is said that the Hamiltonian function is invariant under the canonical transformation, or that the canonical transformation leaves the Hamiltonian invariant, when the new Hamiltonian is the old Hamiltonian in the new coordinates, i.e., 
\begin{equation}
K(q,p,t)=H(Q,P,t),
\end{equation}
or equivalently,
\begin{equation}
H(Q,P,t)=H(q,p,t)+\frac{\partial F}{\partial t}.
\end{equation}
For a time-independent canonical transformation the Hamiltonian is invariant if and only if $H(Q,P,t)=H(q,p,t)$. 

Under the infinitesimal canonical transformation (\ref{eqgeninfcan}), the Hamiltonian is invariant if and only if
\begin{equation}
H(q^{i}+\epsilon\frac{\partial G(q,p,t)}{\partial p_{i}},p_{i}-\epsilon \frac{\partial G(q,p,t)}{\partial q^{i}},t)\ = \ H(q,p,t)\,+\,\epsilon\,\frac{\partial G}{\partial t}\ .
\end{equation}
By expanding $H(q^{i}+\epsilon\frac{\partial G(q,p,t)}{\partial p_{i}},p_{i}-\epsilon \frac{\partial G(q,p,t)}{\partial q^{i}},t)$ about the point $(q,p,t)$ we have
\begin{equation}
H(q,p,t)\ + \ \epsilon\,\frac{\partial H}{\partial q^{i}}\frac{\partial G}{\partial p_{i}}\ - \ \epsilon\frac{\partial H}{\partial p_{i}}\frac{\partial G}{\partial q^{i}} \ + \  {\cal O}(\epsilon^2) \ = \ H(q,p,t)\ +\ \epsilon\,\frac{\partial G}{\partial t}\ ,
\end{equation}
which by disregarding the higher terms ${\cal O}(\epsilon^2)$ in $\epsilon$ gives
\begin{equation}
\frac{\partial G}{\partial t}+\lbrace G,H\rbrace=0\ ,
\end{equation}
thus, $G$ is a constant of motion. In summary, the generating function of an explicit infinitesimal canonical transformation leaving the Hamiltonian invariant is a constant of motion and, reciprocally, the infinitesimal canonical transformation generated explicitly by a constant of motion leaves the Hamiltonian invariant.

\section{Canonical transformations: the geometrical viewpoint}
\label{seccanonicalgeo}

\subsection{The geometric formalism of Hamiltonian Mechanics}
Let us see a brief review of the geometric formulations of time-independent and time dependent Hamiltonian mechanics. The theory of time-independent conservative Hamiltonian systems is naturally constructed within the mathematical formalism of symplectic geometry (for details see \cite{AMRC2008,LR89,Torres2020,Lee2012,marsden2013introduction}). A symplectic manifold is a $2n$ dimensional smooth manifold $M$ equipped with a closed non-degenerate 2-form $\omega$ called a symplectic structure on $M$. Around any point $x\in M$ there exist local coordinates $(q^{1},\ldots,q^{n},p_{1},\ldots,p_{n})$, called canonical coordinates, such that
\begin{equation}
\omega=dq^{i}\wedge dp_{i}\ .
\end{equation}
In this paper we adopt the Einstein summation convention, namely, a summation over repeated indices is assumed.

The Hamiltonian vector field $X_{f}$ for $f\in C^{\infty}(M)$ is given by
\begin{equation}
X_{f}\lrcorner \omega \ = \ df \ .
\end{equation}
The Poisson bracket on $C^{\infty}(M)$ takes the form
\begin{equation}
\lbrace f,g\rbrace= X_{g}f.
\end{equation}
In canonical coordinates $(q^{1},\ldots,q^{n},p_{1},\ldots,p_{n})$, we have
\begin{equation}
X_{f}=\frac{\partial f}{\partial p_{i}}\frac{\partial}{\partial q^{i}}-\frac{\partial f}{\partial q^{i}}\frac{\partial}{\partial p_{i}}
\end{equation}
and
\begin{equation}
\lbrace f,g\rbrace \ = \ \frac{\partial f}{\partial q^{i}}\frac{\partial g}{\partial p_{i}}\,-\,\frac{\partial f}{\partial p_{i}}\frac{\partial g}{\partial q^{i}} \ .
\end{equation}

The Hamiltonian dynamics is defined by the Hamiltonian vector field $X_{H}$ for the given Hamiltonian function $H\in C^{\infty}(M)$ and the equations of motion are the well know Hamilton equations of motion. The evolution (the temporal evolution) of a function $f\in C^{\infty}(M)$ (a physical observable) along the trajectories of the system is given by
\begin{equation}
\dot{f}=X_{H}f=\lbrace f,H\rbrace \ ,
\end{equation}
so that $f$ is a constant of motion if and only if $\lbrace f,H\rbrace=0$. It is clear that the Hamiltonian function is a constant of motion, it represents the total (kinetic plus potential) energy of the system.

On the other hand, the phase space of a time-dependent Hamiltonian system is geometrically identified with the extended phase space $M\times\mathbb{R}$ \cite{struckmeier2005hamiltonian,Torres2020}, where $M$ is a symplectic manifold. The extended phase space $M\times\mathbb{R}$ has a natural structure of cosymplectic manifold (for a brief review of the formalism of cosymplectic geometry see \cite{AE2023,azuaje2023scaling} and for a more detailed exposition see \cite{CLL92,LS2017}). Let $\Omega$ be the symplectic structure on $M$, then locally we have
\begin{equation}
\Omega=dq^{i}\wedge dp_{i}\hspace{1cm}.
\end{equation}
The Hamiltonian vector field $X_{f}$ for $F\in C^{\infty}(M\times\mathbb{R})$ is given by
\begin{equation}
X_{f}\lrcorner \Omega =df-\frac{\partial f}{\partial t}dt \hspace{1cm}\rm{and}\hspace{1cm} X_{f}\lrcorner dt =0,
\end{equation}
where $t$ is the global coordinate on $\mathbb{R}$. In canonical coordinates we have
\begin{equation}
X_{f}=\frac{\partial f}{\partial p_{i}}\frac{\partial}{\partial q^{i}}-\frac{\partial f}{\partial q^{i}}\frac{\partial}{\partial p_{i}}.
\end{equation}

The Poisson bracket is defined in the same way as in the symplectic formalism and, of course, it has the same local form.

The Hamiltonian dynamics is defined by the evolution vector field $E_{H}=X_{H}+\frac{\partial}{\partial t}$ for the given Hamiltonian function $H\in C^{\infty}(M\times\mathbb{R})$. Again, the equations of motion are the Hamilton's equations of motion and the temporal parameter (the time) is the coordinate $t$ on $\mathbb{R}$. The evolution of a function $f\in C^{\infty}(M\times\mathbb{R})$ along the trajectories of the system is given by
\begin{equation}
\dot{f}=E_{H}f=X_{H}f+\frac{\partial f}{\partial t}=\lbrace f,H\rbrace+\frac{\partial f}{\partial t}\ .
\end{equation}
So that $f$ is a constant of motion if and only if $\lbrace f,H\rbrace+\frac{\partial f}{\partial t}=0$. Cosymplectic Hamiltonian systems describe dissipative systems, since the Hamiltonian function $H$ (representing the total energy of the system) is not a constant of motion, $\dot{H}=\frac{\partial H}{\partial t}$ is not zero for time-dependent Hamiltonian functions.

\subsection{Canonical transformations in the geometric time-independent formalism}
Let $(M,\omega,H)$ be a time-independent Hamiltonian system.
\begin{de}
Let $(M,\omega)$ and $(N,\eta)$ two symplectic manifolds. A diffeomorphism $\Phi:M\longrightarrow N$ such that $\Phi^{*}\eta=\omega$ is called a symplectomorphism.
\end{de}
Symplectomorphisms on the phase space of a Hamiltonian system are called (geometric) canonical transformations \cite{CFR2013,AE2023,azuaje2023scaling,struckmeier2005hamiltonian}; explicitly:
\begin{de}
\label{degeocanonical}
A geometric canonical transformation of $(M,\omega,H)$ is a diffeomorphism $\Phi$ on $M$ such that $\Phi^{*}\omega=\omega$.
\end{de}
Observe that as in the classical coordinate-based viewpoint, the geometric canonical transformations are independent of the Hamiltonian function. For the geometric definition of canonical transformations it is only required the existence of a symplectic structure on $M$.

Given a diffeomorphism $\Phi:M\longrightarrow M$,  we know that $\Phi^{*}\omega$ is a symplectic structure on $M$, so around any point $p\in M$ there are local coordinates $(Q^{1},\cdots,Q^{n},P_{1},\cdots,P_{n})$ such that
\begin{equation}
\Phi^{*}\omega=dQ^{i}\wedge dP_{i}.
\end{equation} 
Let us consider the Poisson bracket defined by $\Phi^{*}\omega$, denoted by $\overline{\lbrace,\rbrace}$, we have
\begin{equation}
\overline{\lbrace f,g\rbrace}=\overline{X}_{g}f,
\end{equation}
where $\overline{X}_{g}$ is the Hamiltonian vector field defined by $\Phi^{*}\omega$ for $g\in C^{\infty}(M)$ ($\overline{X}_{g}\lrcorner \Phi^{*}\omega=dg$); in the canonical coordinates $(Q^{1},\cdots,Q^{n},P_{1},\cdots,P_{n})$ it has the form
\begin{equation}
\overline{\lbrace f,g\rbrace}=\frac{\partial f}{\partial Q^{i}}\frac{\partial g}{\partial P_{i}}-\frac{\partial f}{\partial P_{i}}\frac{\partial g}{\partial Q^{i}}.
\end{equation}
If $\Phi$ is a canonical transformation then $\overline{X}_{g}=X_{g}$, so 
\begin{equation}
\overline{\lbrace f,g\rbrace}=\lbrace f,g\rbrace,
\end{equation}
therefore we have
\begin{equation}
\lbrace Q^{i},Q^{j}\rbrace=\overline{\lbrace Q^{i},Q^{j}\rbrace}=0, \quad \lbrace P_{i},P_{j}\rbrace=\overline{\lbrace P_{i},P_{j}\rbrace}=0 \quad {\rm and}\quad \lbrace Q^{i},P_{j}\rbrace=\overline{\lbrace Q^{i},P_{j}\rbrace}=\delta^{i}_{j};
\end{equation}
i.e., if we think of $\Phi$ locally as a coordinate transformation $(q^{1},\cdots,q^{n},p_{1},\cdots,p_{n})\mapsto (Q^{1},\cdots,Q^{n},P_{1},\cdots,P_{n})$, 
then it is a canonical transformation in the classical sense. Sometimes we will ad the adjective ¨geometric¨ to canonical transformations given in the geometrical sense in order to avoid confusions.

\subsection{Canonical transformations in the geometric time-dependent formalism}
Let $(M\times\mathbb{R},\Omega,H)$ be a time-dependent Hamiltonian system, i.e., $H\in C^{\infty}(M\times\mathbb{R})$. For any transformation (diffeomorphism) $\Phi$ on $M\times\mathbb{R}$ we require that it preserves the temporal parameter $t$, which implies that we can consider the map $\Phi_{t}:M\longrightarrow M$ for each fixed value of $t$ and it is a transformation, i.e., a diffeomorphism. We propose the following definition of canonical transformation within the geometric formalism of time-dependent Hamiltonian Mechanics.

\begin{de}
\label{decanonicaltime}
A canonical transformation of $(M\times\mathbb{R},\Omega,H)$ is a diffeomorphism $\Phi$ on $M\times\mathbb{R}$ such that 
\begin{equation}
\label{eqtimecanonical}
\Phi^{*}\Omega-\Omega=d(H-K)\wedge dt,
\end{equation}
for some function $K\in C^{\infty}(M\times\mathbb{R})$.
\end{de}
The above definition (\ref{eqtimecanonical}) is more general than the one previously 
presented in \cite{azuaje2023scaling}.

Now, one could think that definition (\ref{eqtimecanonical}) of canonical transformations is not independent of the Hamiltonian $H$. However, equation (\ref{eqtimecanonical}) just tells us the relation between two Hamiltonians $H$ and $K$. In fact. we can rewrite this definition without considering a Hamiltonian system at all, but only a phase space of the form $(M\times\mathbb{R},\Omega)$. An equivalent definition to (\ref{eqtimecanonical}) can be stated as follows: a canonical transformation of $(M\times\mathbb{R},\Omega)$ is a diffeomorphism $\Phi$ on $M\times\mathbb{R}$ (preserving $t$) such that 
\begin{equation}
\Phi^{*}\Omega-\Omega=dJ\wedge dt,
\end{equation}
for some function $J\in C^{\infty}(M\times\mathbb{R})$. For a given Hamiltonian function $H$, the new Hamiltonian $K$ satisfies $d(H-K)=dJ$.

We can observe that locally equation (\ref{eqtimecanonical}) is of the form
\begin{equation}
dQ^{i}\wedge dP_{i}-dq^{i}\wedge dp_{i}=d(H-K)\wedge dt,
\end{equation}
which is the general condition for a coordinate canonical transformation with new Hamiltonian function $K$ \cite{Calkin96,torres2018introduction}.

\subsection{One-parameter groups of canonical transformations and its generators}
\label{subseconeparameter}
In this subsection $M$ may be a symplectic manifold or $M\times\mathbb{R}$, unless it is explicitly indicated.
\begin{de}
\label{delog}
A local one-parameter group of transformations on $M$ is a family of diffeomorphisms $\lbrace \Psi_{s}\rbrace_{s\in I}$ on $M$ with $I$ an open interval containing the zero, such that for each $s$, $\Psi_{0}=Id_{M}$ (the identity function), $\Psi_{s}\circ \Psi_{t}=\Psi_{s+t}$ (the group property) and the map $\Psi:I\times M\longrightarrow M$ defined by $\Psi(s,p)=\Psi_{s}(p)$ is a smooth map. 
\end{de}
Definition \ref{delog} can be rewritten by just saying that a local one-parameter group of transformations $\lbrace \Psi_{s}\rbrace_{s\in I}$ on $M$ is a smooth local flow $\Psi:I\times M\longrightarrow M$ on $M$ \cite{Lee2012}. The interval for the parameter $s$ may be the whole set of real numbers, in this case the family of transformations is called a global one-parameter group of transformations and the map $\Psi:\mathbb{R}\times M\longrightarrow M$ is a smooth global flow.

Given a smooth flow $\Psi$ (local or global) on $M$, for each $p\in M$ we can consider the smooth curve $\Psi^{p}:I\longrightarrow M$ defined by $\Psi^{p}(s)=\Psi(s,p)$, and the tangent vector $v_{p}\in T_{p}M$ given by $v_{p}=\frac{d}{ds}\Psi^{p}(s)|_{s=0}$; we have that the assignment $p\mapsto v_{p}$ is a smooth vector field $V$ on $M$ called the infinitesimal generator of the flow . The fundamental theorem of flows states (in few words) that each smooth vector field $V$ on $M$ is the infinitesimal generator of a (possibly local) smooth flow $\Psi$ such that $\Psi_{s}:M\longrightarrow M$ is a (possibly local) diffeomorphism and $\Psi^{p}:I\longrightarrow M$ is an integral curve of $V$ \cite{Lee2012}. We can make reference to the flow $\Psi$ or to the one-parameter group of transformations $\lbrace\Psi_{s}\rbrace_{s\in I}$ since they are in a one-to-one correspondence and, we can refer to the infinitesimal generator of the flow $\Psi$ as the infinitesimal generator of the one-parameter group of transformations $\lbrace\Psi_{s}\rbrace_{s\in I}$. 

\begin{de}
A (local) one-parameter group of canonical transformations is a (local) one-parameter group of transformations  $\lbrace \Psi_{s}\rbrace_{s\in I}$ on $M$ such that $\Psi_{s}:M\longrightarrow M$ is a canonical transformation for each $s$. 
\end{de}
The following theorem tell us that infinitesimal generators of one-parameter groups of canonical transformations are Hamiltonian vector fields.
\begin{te}
Let  $\lbrace \Psi_{s}\rbrace_{s\in I}$ be a one-parameter group of transformations on $M$ and $V\in \mathfrak{X}(M)$ its infinitesimal generator. $\lbrace \Psi_{s}\rbrace_{s\in I}$ is a one-parameter group of canonical transformations if and only if $V$ is a Hamiltonian vector field. 
\end{te}
\begin{proof}
First let us see the time-independent case. $\lbrace \Psi_{s}\rbrace_{s\in I}$ is a one-parameter group of canonical transformations if and only if $\Psi_{s}^{*}\omega=\omega$ for every $s\in I$, which is equivalent to $L_{V}\omega=0$. We have
\begin{equation}
L_{V}\omega=V\lrcorner d\omega+d(V\lrcorner\omega)=d(V\lrcorner\omega),
\end{equation}
so 
\begin{equation}
L_{V}\omega=0 \Longleftrightarrow d(V\lrcorner\omega)=0.
\end{equation}
Equation $d(V\lrcorner\omega)=0$ is equivalent to the existence of a possibly-local function $f$ on $M$ such that $V\lrcorner\omega=df$, which means that $V$ is a Hamiltonian vector field.

Now let us aboard the time-dependent case, since the transformation $\Psi_{s}$ preserve the temporal parameter $t$, then we have $V\lrcorner dt=0$, i.e., for each fixed value of $t$ we can consider the vector field $V|_{M}$ on $M$, and we have that it is the infinitesimal generator of a one-parameter group of canonical transformations on $M$, so we must have $\Psi_{t,s}^{*}\Omega=\Omega$ for each fixed value of $t$, then by the previous calculus we have that on $M$ for each fixed value of $t$ there exists a function $f_{t}$ such that
\begin{equation}
V\lrcorner\Omega=df_{t}.
\end{equation}
If we consider the $f=f_{t}$ as a differentiable function on $M\times\mathbb{R}$ then we have
\begin{equation}
V\lrcorner\Omega=df-\frac{\partial f}{\partial t}dt,
\end{equation}
which means that $V$ is a Hamiltonian vector field.
\end{proof}
For a one-parameter group of canonical transformations with infinitesimal generator $V=X_{f}$ we say that $f$ is the generator of such one-parameter group of canonical transformations.

We have seen that (geometric) canonical transformations $\Psi$ can be seen locally as coordinate canonical transformations $(q^{1},\cdots,q^{n},p_{1},\cdots,p_{n})\mapsto (Q^{1},\cdots,Q^{n},P_{1},\cdots,P_{n})$, so a one-parameter group of canonical transformations can be seen locally as a family of coordinate canonical transformations of the form $\lbrace (Q_{s}(q,p,t),P_{s}(q,p,t))\rbrace_{s\in I}$. 

One-parameter groups of coordinate canonical transformations of the form $\lbrace (Q_{s}(q,p,t),P_{s}(q,p,t))\rbrace_{s\in I}$ are studied in \cite{torres2018introduction}, there the authors use the notation $(Q(q,p,t,s),P(q,p,t,s))$ for the family $\lbrace (Q_{s}(q,p,t),P_{s}(q,p,t))\rbrace_{s\in I}$ and it is shown that if $(Q(q,p,t,s),P(q,p,t,s))$ is a one-parameter group of coordinate canonical transformations then there exists a function $f(q,p,t)$ (a possibly local function on $M\times\mathbb{R}$) such that
\begin{equation}
\label{eqgengroup}
\frac{\partial Q^{i}}{\partial s}|_{s=0}=\frac{\partial f}{\partial p_{i}}\quad {\rm and}\quad \frac{\partial Q^{i}}{\partial s}|_{s=0}=-\frac{\partial f}{\partial q^{i}}.
\end{equation}
And reciprocally, any differentiable function $f(q,p,t)$ defines a one-parameter group of coordinate canonical transformations of the form $(Q(q,p,t,s),P(q,p,t,s))$ such that equations (\ref{eqgengroup}) hold. It is easy to see that if we have a one-parameter group of (geometric) canonical transformations $\Psi$ generated by a function $f\in C^{\infty}(M\times\mathbb{R})$ (the infinitesimal generator is the Hamiltonian vector field $X_{f}$) then the generator of the one-parameter group of coordinate canonical transformations $(Q(q,p,t,s),P(q,p,t,s))$ defined locally by $\Psi$ ($\Psi$ seen as a flow is a possibly time-dependent vector function depending on the points $(q,p)\in M$ and the parameter $s$) is $f(q,p,t)$. Specifically, equations (\ref{eqgengroup}) are the Hamilton's equations of motion for the vector field $X_{f}$, so the solutions of such equations are the integral curves of $X_{f}$ which together define the flow $\Psi$ of $X_{f}$. An ilustraation resuming the main elements in the construction of one-parameter groups of canonical transformations as flows of Hamiltonian vector fields can be seen in Figure \ref{fig1}.

\begin{figure}[ht]
\centering
\includegraphics[width=12cm]{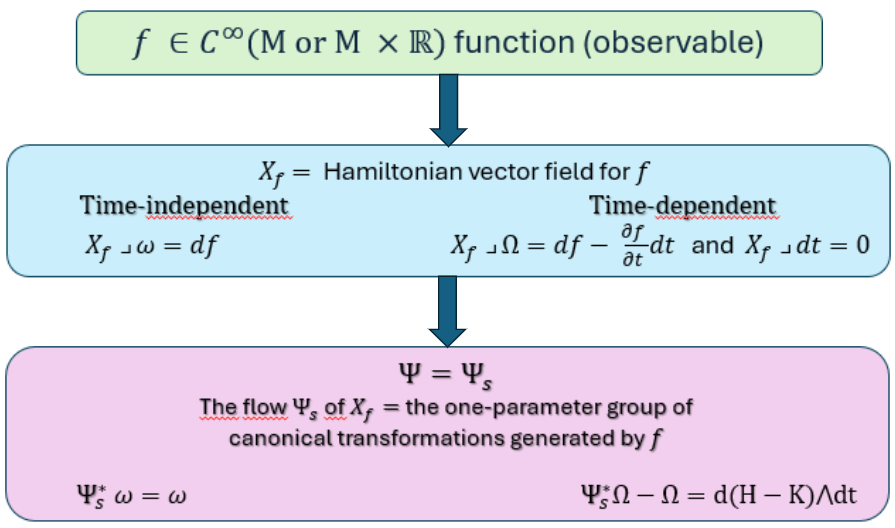}
\caption{A one-parameter group of canonical transformations represented as the flow of a Hamiltonian vector field.}
\label{fig1}
\end{figure}

Below we present an example to illustrate the ideas discussed previously. 
\begin{ex}
Let us consider the one-parameter group of coordinate canonical transformations given by 
\begin{equation}
\label{eqoneparacoord}
Q=qe^{s}-\frac{tp}{m}(e^{s}-e^{-s}), \quad P=pe^{-s}.
\end{equation}
In \cite{torres2018introduction} it is shown that this one-parameter group of coordinate canonical transformations is generated by $f(q,p,t)=qp-\frac{tp^{2}}{m}$, i.e., equations (\ref{eqgengroup}) hold. Firstly, let us observe that the flow of $X_{f}$ is given locally by $\Psi(q,p,t,s)=(qe^{s}-\frac{tp}{m}(e^{s}-e^{-s}),pe^{-s})$. Indeed, the equations for the integral curves $\Psi(s)=(q(s),p(s))$ of $X_{f}$ read
\begin{equation}
\dot{q}=q-\frac{2tp}{m},\quad \dot{p}=-p\ ,
\end{equation}
with the general solution 
\begin{equation}
q(s)=e^{s}c_{1}-\frac{(e^{s}-e^{-s})t}{m}c_{2},\quad p(s)=e^{-s}c_{2}\ ,
\end{equation}
the constants $c_{1},c_{2}$ are determined by $q(0)=q,p(0)=p$, i.e., we have $c_{1}=q$ and $c_{2}=p$. So the flow of $X_{f}$ is given locally by
\begin{equation}
\Psi(q,p,t,s)\ = \ (q\,e^{s}-\frac{t\,p}{m}(e^{s}-e^{-s}),\,pe^{-s}).
\end{equation}
It follows that the one-parameter group of coordinate canonical transformations $(Q(q,p,t,s)=qe^{s}-\frac{tp}{m}(e^{s}-e^{-s}),P(q,p,t,s)=pe^{-s})$ is the local expression for the one-parameter group of (geometric) canonical transformations $\Psi$ with infinitesimal generator $X_{f}$ (or generator $f$).

We can verify that $\Psi$ is a bona fide one-parameter group of geometric canonical transformation, indeed
\begin{equation}
\begin{split}
dQ\wedge dP-dq\wedge dp &=(e^{s}dq-\frac{t}{m}(e^{s}-e^{-s})dp-\frac{p}{m}(e^{s}-e^{-s})dt)\wedge(e^{-s}dp)-dq\wedge dp\\
&=-\frac{p}{m}(e^{s}-e^{-s})dt\wedge dp\\
&=\frac{p}{m}(e^{s}-e^{-s})dp\wedge dt.
\end{split}
\end{equation}
Hence, $\Psi$ is a one-parameter group of canonical transformation. Observe that up to now, we have not considered any Hamiltonian function (in complete agreement with the fact that canonical transformations are independent of the Hamiltonian function, or in other words, canonical transformations are independent of the dynamics). Now, given a Hamiltonian $H$, the new Hamiltonian $K$ reads $K(q,p,t,s)=H-\frac{p}{m}(e^{s}-e^{-s})+g(t,s)$,
where $g(t,s)$ is any function of $t$ and $s$. 

We can also verify that for each fixed value of $t$ we have $\Psi_{t,s}^{*}\Omega=\Omega$. Explicitly,
\begin{equation}
\begin{split}
\Psi_{t,s}^{*}\Omega &=dQ\wedge dP=(e^{s}dq-\frac{t}{m}(e^{s}-e^{-s})dp)\wedge e^{-s}dp\\
&=dq\wedge dp\\
&=\Omega, 
\end{split}
\end{equation}
where the differentials are on $M$ so the variables $t,s$ are treated as parameters.
\end{ex}

\textbf{Remark}: Infinitesimal generators of one-parameter groups of canonical transformations are not considered in \cite{torres2018introduction} and, to the best of our knowledge, in the literature there is no description of them in the time-dependent framework.

Infinitesimal canonical transformations are the infinitesimal version of one-parameter groups of coordinate canonical transformations of the form $(Q(q,p,t,s),P(q,p,t,s))$, namely, by expanding about $s=0$ and neglecting second and higher order terms in $s$ we obtain the infinitesimal version of the one-parameter group of coordinate canonical transformations (we take $\epsilon=s$ to be the infinitesimal parameter). The infinitesimal version of the one-parameter group of coordinate canonical transformations $(Q(q,p,t,s)=qe^{s}-\frac{tp}{m}(e^{s}-e^{-s}),P(q,p,t,s)=pe^{-s})$ is presented in \cite{torres2018introduction} and it is
\begin{equation}
\label{eqinfinitesimal}
Q=q+q\,\epsilon-\frac{2\,t\,p}{m}\epsilon,\quad P=p-p\,\epsilon.
\end{equation}
Clearly, the generating function of an infinitesimal canonical transformation is the same generating function of the one-parameter group of coordinate canonical transformations defining it. For the infinitesimal canonical transformation (\ref{eqinfinitesimal}) the generating function is $G(q,p,t)=qp-\frac{tp^{2}}{m}$ which coincides with the generating function $f(q,p,t)$ of the one-parameter group of coordinate canonical transformations (\ref{eqoneparacoord}). In general, as a natural fact, the infinitesimal canonical transformation with a given generating function is the infinitesimal version of the one-parameter group of coordinate canonical transformations generated by the same function.

\section{Symmetries of Hamiltonian systems and groups of invariance transformations}
\label{subsecsymmetries}
It is well known that associated with any infinitesimal invariance transformation there is a constant of motion, which is the generator of the transformation \cite{Calkin96} (see subsection \ref{subsec2.3} for a brief review). In addition, in \cite{torres2018introduction} it is shown that if a one-parameter group of coordinate canonical transformations leaves the Hamiltonian invariant, then its generating function is a constant of motion; and reciprocally, any constant of motion is the generating function of a one-parameter group of coordinate canonical transformations leaving the Hamiltonian invariant. It is said that a Hamiltonian $H(q,p,t)$ is invariant under a coordinate canonical transformation $(Q(q,p,t),P(p,p,t))$ (or that it leaves the Hamiltonian invariant) if the new Hamiltonian is
\begin{equation}
\label{eqinvariance}
K(q,p,t)\ = \ H(Q(q,p,t),P(p,p,t),t) \ .
\end{equation}
If the canonical transformation is time-independent, i.e., $Q=Q(q,p),P=P(q,p)$, then the new Hamiltonian can be taken the same as the old one $K(q,p,t)=H(q,p,t)$. So a time-independent canonical transformation $(Q(q,p),P(q,p))$ leaves the Hamiltonian $H(q,p,t)$ invariant if $H(Q,P,t)=H(q,p,t)$.

From the previous discussion, we can see that for time-independent transformations the condition for invariance can be more specific or restrictive than in the time-dependent case; of course condition 
(\ref{eqinvariance}) defines invariance in both cases, time-independent and time-dependent.

In this section we develop the invariance notion from the geometric framework, we show that this geometric notion coincides locally with the invariance notion in the coordinate-based approach and, we relate the notion of invariance with the notion of infinitesimal symmetry. In the geometric formalism we can also consider a more specific condition in the time-independent case. Moreover, as we will see soon, there is a general condition for invariance containing both the time-independent and the time-dependent cases. In order to make the developments in this paper as much clear as possible, we present first the geometric formalism for invariance in the time-independent case.

Let $(M,\omega,H)$ be a time-independent Hamiltonian system.
\begin{de}
It is said that the time-independent canonical transformation  $\Phi:M\longrightarrow M$ leaves the Hamiltonian $H$ invariant if $\Phi^{*}H=H$ (or equivalently $H\circ \Phi=H$).
\end{de}
In the case of a one-parameter group of time-independent canonical transformations $\lbrace\Psi_{s}\rbrace$, it is said that the Hamiltonian $H$ is invariant under the group if it is invariant under each transformation $\Psi_{s}$, i.e., $\Psi_{s}^{*}H=H$.

It is clear that the infinitesimal generator of a one-parameter group of time-independent canonical transformations leaving the Hamiltonian invariant is a Hamiltonian vector field $X_{f}$ such that
$L_{X_{f}}\omega=L_{X_{f}}H=0$. A vector field $X$ satisfying
\begin{equation}
L_{X}\omega=L_{X}H=0
\end{equation}
is called an infinitesimal symmetry of the Hamiltonian system $(M,\omega,H)$ \cite{Lee2012}. One of the
consequences of the existence of symmetries in general in Classical Mechanics is their association with constants of motion, this is the result of the celebrated Noether theorem. There is a version of Noether theorem for Lagrangian Mechanics \cite{Arnold78,KS2011} as well as for Hamiltonian Mechanics \cite{torres2018introduction} without the employment of geometric formalisms, relating one-parameter groups of coordinate canonical transformations (leaving the Hamiltonian invariant) with constants of motion. Also, the geometric version for Hamiltonian Mechanics reads \cite{Lee2012,azuaje2023scaling,BG2021}:
\begin{te}
\begin{enumerate}
\item[i)] If $f$ is a constant of motion of $(M,\omega,H)$, then its associated Hamiltonian vector field 
$X_{f}$ is an infinitesimal symmetry of $(M,\omega,H)$. 
\end{enumerate}
Reciprocally,
\begin{enumerate}
\item[ii)] if $V$ is an infinitesimal symmetry of $(M,\omega,H)$, then there is a (possibly-local) function $f$ such that $V=X_{f}$ and $f$ is a constant of motion.
\end{enumerate}
\end{te}
We conclude that if a one-parameter group of time-independent (geometric) canonical transformations leaves the Hamiltonian invariant, then its generating function is a constant of motion; and reciprocally, any time-independent constant of motion is the generating function of a one-parameter group of time-independent (geometric) canonical transformations leaving the Hamiltonian invariant.

Now we address the problem of invariance under one-parameter groups of time-dependent (geometric) canonical transformations. Let $(M\times\mathbb{R},\Omega,H)$ be a time-dependent Hamiltonian system. We propose the following notion of invariance from the geometric approach.
\begin{de}
We say that the canonical transformation $\Psi:M\times\mathbb{R}\longrightarrow M\times\mathbb{R}$ leaves the Hamiltonian invariant if 
\begin{equation}
\label{eqtimeinvariance}
\Psi^{*}(\Omega+dH\wedge dt)=\Omega+dH\wedge dt.
\end{equation}
\end{de}
Equation (\ref{eqtimeinvariance}) is equivalent to
\begin{equation}
\Psi^{*}\Omega+d(\Phi^{*}H)\wedge dt=\Omega+dH\wedge dt,
\end{equation}
which can be rewritten as
\begin{equation}
\Psi^{*}\Omega-\Omega=d(H-\Phi^{*}H)\wedge dt;
\end{equation}
this last equation implies that the new Hamiltonian is
\begin{equation}
K=\Phi^{*}H=H\circ\Phi.
\end{equation}
Locally it can be written as
\begin{equation}
K(q,p,t)=H(Q,P,t),
\end{equation}
which is the general condition for a coordinate canonical transformation (time-independent or time-dependent) to be an invariance transformation. 

In the case of a one-parameter group of time-dependent canonical transformations $\lbrace\Psi_{s}\rbrace$, we say that the Hamiltonian is invariant under the group if it is invariant under each transformation $\Psi_{s}$, i.e.,
\begin{equation}
\Psi_{s}^{*}(\Omega+dH\wedge dt)=\Omega+dH\wedge dt.
\end{equation}

It is clear that the infinitesimal generator of a one-parameter group of time-dependent canonical transformations leaving the Hamiltonian invariant is a Hamiltonian vector field $X_{f}$ such that
$L_{X_{f}}(\Omega+dH\wedge dt)=0$. In this paper a vector field $X$ satisfying
\begin{equation}
L_{X}(\Omega+dH\wedge dt)=0
\end{equation}
shall be called an infinitesimal symmetry of the Hamiltonian system $(M\times\mathbb{R},\Omega,H)$. In \cite{azuaje2023scaling} it is presented, within the cosymplectic framework, a different definition for infinitesimal symmetries and a version of the Noether theorem relating infinitesimal symmetries and constants of motion. Here we present a version of the Noether theorem relating infinitesimal symmetries and constants of motion; this theorem is the geometric version of the one presented in \cite{torres2018introduction} for one-parameter groups of coordinate canonical transformations leaving the Hamiltonian invariant.
\begin{te}
\begin{enumerate}
\item[i)] If $f$ is a constant of motion of $(M\times\mathbb{R},\Omega,H)$, then its associated Hamiltonian vector field 
$X_{f}$ is an infinitesimal symmetry of $(M\times\mathbb{R},\Omega,H)$. 
\end{enumerate}
Reciprocally,
\begin{enumerate}
\item[ii)] if $V$ is an infinitesimal symmetry of $(M\times\mathbb{R},\Omega,H)$, then there is a (possibly-local) function $f$ such that $V=X_{f}$ and $f$ is a constant of motion.
\end{enumerate}
\end{te}
\begin{proof}
Let $g\in C^{\infty}(M\times\mathbb{R})$. We have
\begin{equation}
\begin{split}
L_{X_{g}}(\Omega+dH\wedge dt) &= L_{X_{g}}\Omega+d(L_{X_{g}}H)\wedge dt+dH\wedge d(L_{X_{g}}t)\\
&= d(X_{g}\lrcorner\Omega)+d\lbrace H,g\rbrace\wedge dt+dH\wedge d\lbrace t,g\rbrace\\
&= d(dg-\frac{\partial g}{\partial t}dt)+d\lbrace H,g\rbrace\wedge dt\\
&= -d(\frac{\partial g}{\partial t})\wedge dt-d\lbrace H,g\rbrace\wedge dt\\
&= -d(\lbrace H,g\rbrace+\frac{\partial g}{\partial t})\wedge dt,
\end{split}
\end{equation}
so we have that
\begin{equation}
L_{X_{g}}(\Omega+dH\wedge dt)=0 \Longleftrightarrow d(\lbrace H,g\rbrace+\frac{\partial f}{\partial t})=0 \Longleftrightarrow \lbrace H,g\rbrace+\frac{\partial g}{\partial t}=c \ ,
\end{equation}
where $c$ is a real constant. If $g$ is such that $\lbrace H,g\rbrace+\frac{\partial g}{\partial t}=c$ We can always chose $f=g-ct$, then $X_{f}=X_{g}$ and
\begin{equation}
\lbrace H,f\rbrace+\frac{\partial f}{\partial t}=\lbrace H,g\rbrace+\frac{\partial g}{\partial t}-c\ = \ 0 \ ,
\end{equation}
i.e., $X_{f}$ is an infinitesimal symmetry if and only if $f$ is a constant of motion.
\end{proof}
Below, we graphically resume in Figure \ref{fig2} the main results of this section \ref{subsecsymmetries}.
\begin{figure}[ht]
\centering
\includegraphics[width=12cm]{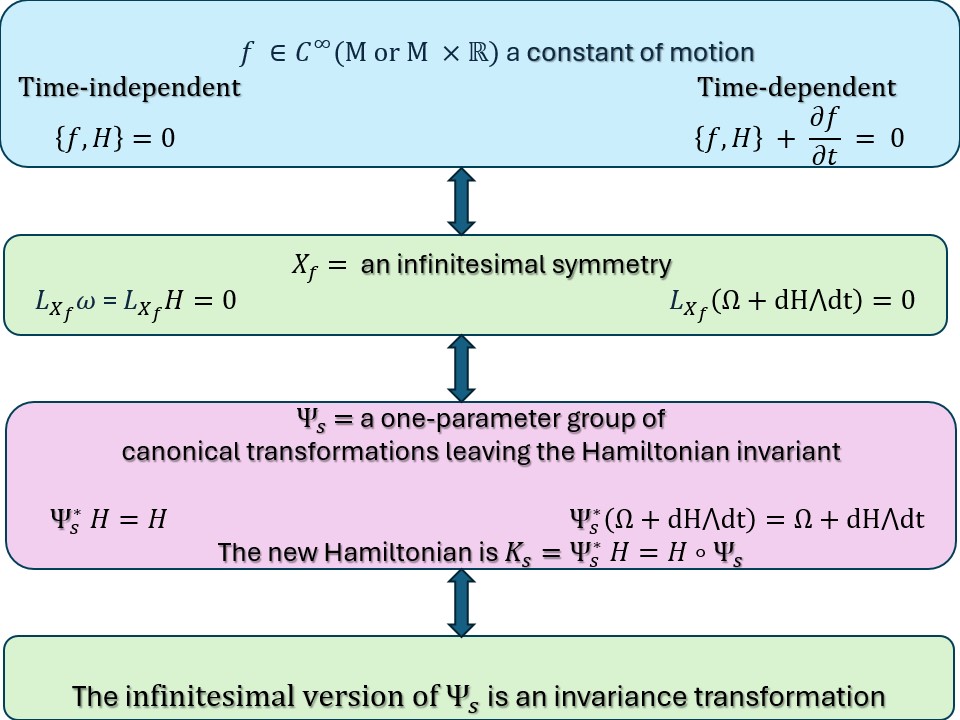}
\caption{Equivalence between the existence of, a constant of motion, an infinitesimal symmetry, a one-parameter group of canonical transformations leaving the Hamiltonian invariant and, an infinitesimal invariance transformation.}
\label{fig2}
\end{figure}

\section{Examples}
\label{secexample}

\subsection{A one-parameter group of canonical transformations for the isotropic harmonic oscillator}
Let us consider the physically relevant isotropic harmonic oscillator in 2D. The Hamiltonian function in canonical Cartesian coordinates $(x,y,p_{x},p_{y})$ is 
\begin{equation}
H=\frac{1}{2\,m}(p_{x}^{2}+p_{y}^{2})+\frac{m\, w^{2}(x^{2} +y^{2})}{2},
\end{equation}
$m$ is the mass and $w$ the angular frequency. The angular momentum
\begin{equation}
f = x\,p_{y} - y\,p_{x} \ ,
\end{equation}
is a constant of motion, so the associated Hamiltonian vector field $X_{f}$ is an infinitesimal symmetry, indeed,
\begin{equation}
X_{f} = -y\frac{\partial}{\partial x} - p_{y}\frac{\partial}{\partial p_{x}} + x\frac{\partial}{\partial y} + p_{x}\frac{\partial}{\partial p_{y}},
\end{equation}
\begin{equation}
\begin{split}
L_{X_{f}}\omega &= d(X_{f}\lrcorner X_{f})\\
&= d((dx\wedge dp_{x}+dy\wedge dp_{y})\lrcorner X_{f})\\
&= d(-ydp_{x}+p_{y}dx+xdp_{y}-p_{x}dy)\\
&= -dy\wedge dp_{x}+dp_{y}\wedge dx+dx\wedge dp_{y}-dp_{x}\wedge dy\\
&= 0 \ ,
\end{split}
\end{equation}
and
\begin{equation}
L_{X_{f}}H=X_{f}H=-y\,m\,w^{2}\,x-p_{y}\frac{p_{x}}{m}+x\,m\,w^{2}\,y+p_{x}\frac{p_{y}}{m}=0.
\end{equation}
Since $X_{f}$ is an infinitesimal symmetry then the flow of $X_{f}$ is a one-parameter group of time-independent canonical transformations leaving the Hamiltonian $H$ invariant, see below. The integral curves of $X_{f}$ satisfy the Hamilton equations for $f$,
\begin{equation}\label{eq76}
\frac{\partial f}{\partial x} = -\dot{p_{x}}, \quad \frac{\partial f}{\partial y} = -\dot{p_{y}}, \quad  \frac{\partial f}{\partial p_{x}}=\dot{x}, \quad \frac{\partial f}{\partial p_{y}} = \dot{y}, 
\end{equation} 
i.e., 
\begin{equation}
\label{eqhamiltonangular}
\begin{array}{c}
p_{y}= -\dot{p_{x}} \\
p_{x}= \dot{p_{y}} \\
y= - \dot{x} \\
x=  \dot{y}\ .
\end{array}
\end{equation}
The general solution of system (\ref{eqhamiltonangular}) is
\begin{equation}\label{eq79}
\begin{split}
y &= c_{1}\cos s + c_{2} \sin s \\
x &= -c_{1}\sin s + c_{2}\cos s \\
p_{y} &= c_{3}\cos s + c_{4}\sin s \\
p_{x} &= -c_{3}\sin s + c_{4}\cos s\ ,
\end{split}
\end{equation}
with $c_{1},c_{2},c_{3},c_{4}$ constants of integration. Since the flow $\Psi(p,s)$ of $X_{f}$ must obey the condition $\Psi(p,0)=p$, we obtain that
\begin{equation}
\Psi(x,y,p_{x},p_{y},s) =(-y\sin s + x\cos s,\,y\cos s + x\sin s,\,p_{x}\cos s - p_{y}\sin s,\,p_{y}\cos s + p_{x}\sin s).
\end{equation}
Now, we can see explicitly that
\begin{equation}
\begin{split}
\Psi_{s}^{*}\omega &= \Psi_{s}^{*}(dx\wedge dp_{x}+dy\wedge dp_{y})\\
&= d(-y\sin s + x\cos s)\wedge d(p_{x}\cos s - p_{y}\sin s)+d(y\cos s + x\sin s)\wedge d(p_{y}\cos s + p_{x}\sin s)\\
&= dx\wedge dp_{x}+dy\wedge dp_{y}\\
&= \omega \ ,
\end{split}
\end{equation}
and
\begin{equation}
\begin{split}
\Psi_{s}^{*}H &= H\circ \Psi_{s}\\
&= H(-y\sin s + x\cos s,y\cos s + x\sin s,\,p_{x}\cos s - p_{y}\sin s,\,p_{y}\cos s + p_{x}\sin s)\\
&= \frac{1}{2\,m}((p_{x}\cos s - p_{y}\sin s)^{2}+(p_{y}\cos s + p_{x}\sin s)^{2})+\frac{m\,w^{2}((-y\sin s + x\cos s)^{2} +(y\cos s + x\sin s)^{2})}{2}\\
&= \frac{1}{2\,m}(p_{x}^{2}+p_{y}^{2})\ + \ \frac{m\,w^{2}(x^{2} +y^{2})}{2}\\
&= H.
\end{split}
\end{equation}

\subsection{A one-parameter group of canonical transformations for the Smorodinsky-Winternitz system}
For this example we consider the superintegrable Smorodinsky-Winternitz system \cite{winternitz1967symmetry,bonatsos1994deformed}. The Hamiltonian function is given by
\begin{equation}
H = \frac{(p_{x}^{2} + p_{y}^{2})}{2} + k\,(x^{2} + y^{2}) + \frac{c}{x^{2}}\ ,
\end{equation}
$k,c$ are constants. The function
\begin{equation}
T \ = \  p_{y}^{2} + 2\,k\,y^{2} \ ,
\end{equation}
is a constant of motion, responsible for the separation of variables in Cartesian coordinates. Therefore the Hamiltonian vector field
\begin{equation}
X_{T} = 2p_{y}\frac{\partial}{\partial y} - 4ky\frac{\partial}{\partial p_{y}} \ ,
\end{equation}
is a infinitesimal symmetry and the flow
\begin{equation}
\Psi(x, y, p_{x}, p_{y}, s) = \left(x, \frac{p_{y}\sin(\sqrt{8k}s)}{\sqrt{2k}} + y\cos(\sqrt{8k}s),p_{x}, p_{y}\cos(\sqrt{8k}s) - \sqrt{2k}y\sin(\sqrt{8k}s) \right) \ ,
\end{equation}
is a one-parameter group of canonical transformations leaving the Hamiltonian invariant. Indeed,
\begin{equation}
\begin{split}
\Psi^{*}\omega= &= \left(\frac{\sin(\sqrt{8k}s)}{\sqrt{2k}}dp_{y} + \cos(\sqrt{8k}s)dy \right) \wedge \\&\quad \left( \cos(\sqrt{8k}s)dp_{y} - \sqrt{2k}\sin(\sqrt{8k}s)dy \right)+ dx \wedge dp_{x}, \\
&= (-\sin^{2}(\sqrt{8k}s)dp_{y} \wedge dy + \cos^{2}(\sqrt{8k}s) dy \wedge dp_{y}) + \\
&\quad + dx \wedge dp_{x}\\
&= dx \wedge dp_{x} + dy \wedge dp_{y}\\
&= \omega \ ,
\end{split}
\end{equation}
and
\begin{equation}
\begin{split}
=H(\Psi(x,y,p_{x},p_{y},s)) &= \frac{1}{2}\left(p_{x}^{2} + \left(p_{y}\cos(\sqrt{8k}s) - \sqrt{2k}y\sin(\sqrt{8k}s)\right)^{2}\right) + \\
&\quad + k\left(x^{2} + \left(\frac{p_{y}\sin(\sqrt{8k}s)}{\sqrt{2k}} + y\cos(\sqrt{8k}s)\right)^{2}\right) + \frac{c}{x^{2}} \\
&= \frac{(p_{x}^{2} + p_{y}^{2})}{2} \ + \  k\,(x^{2} + y^{2}) \  + \  \frac{c}{x^{2}}\\
&= H(x,y,p_{x},p_{y}).
\end{split}
\end{equation}

\subsection{A one-parameter group of invariance  time-dependent canonical transformations}
In \cite{torres2018introduction} is is shown that the time-dependent Hamiltonian
\begin{equation}
H=\frac{p}{2\,m}\ - \ k\,t\,q,
\end{equation}
with $k$ a real constant, is invariant under the one-parameter group of time-dependent coordinate canonical transformations generated by the time-dependent constant of motion 
\begin{equation}
g=q-\frac{t\,p}{m}+k\frac{t^{3}}{3\,m}. 
\end{equation}
We can verify that the Hamiltonian vector field $X_{g}$ of $g$ is an infinitesimal symmetry of the system.
Indeed,
\begin{equation}
X_{g}=-\frac{t}{m}\frac{\partial}{\partial q}-\frac{\partial}{\partial p},
\end{equation}
so we have
\begin{equation}
\begin{split}
L_{X_{g}}(\Omega+dH\wedge dt) &= d(\Omega\lrcorner X_{g}+´dH\wedge dt\lrcorner X_{g})\\
&= d(dq\wedge dp\lrcorner(-\frac{t}{m}\frac{\partial}{\partial q}-\frac{\partial}{\partial p})-ktdq\wedge dt\lrcorner(-\frac{t}{m}\frac{\partial}{\partial q}-\frac{\partial}{\partial p})+\frac{p}{m}dp\wedge dt\lrcorner(-\frac{t}{m}\frac{\partial}{\partial q}-\frac{\partial}{\partial p}))\\
&= d(-\frac{t}{m}dp+dq+k\frac{t^{2}}{m}dt+\frac{p}{m}dt)\\
&= 0.
\end{split}
\end{equation}
The flow $\Psi$ of $g$ is given by
\begin{equation}
\Psi(q,p,t,s)=(q-\frac{ts}{m},p-s,t),
\end{equation}
which coincides with the expression of the one-parameter group of coordinate canonical transformations generated by $g$ given in \cite{torres2018introduction} (according to equation (\ref{eqgengroup})) by
\begin{equation}
Q=q-\frac{ts}{m}\ , \quad P=p-s.
\end{equation}

\section{Conclusions}
Summarizing, a complete description of canonical transformations under the geometric formalism of Hamiltonian systems (on symplectic manifolds as well as on the extended phase space with structure of cosymplectic manifold) is presented. It was demonstrated that the local aspects of the geometric formalism coincide exactly with the coordinate-based classical description. Also, some important well-known results have been geometrically formalized, for instance, the relation between constants of motion and invariance transformations through the identification of groups of invariance transformations as flows of vector fields that are symmetries of the system. In particular, it was shown that these vector fields are in clear correspondence via the Noether theorem with the constants of motion.

The geometric description of canonical transformations described in this paper allows both a formal treatment of groups of canonical transformations as well as their geometric visualization. This is possible thanks to the introduction of vector fields representing symmetries. So we believe that the present results can provide insights, from a physical and a mathematical perspective, into the modern theory of Hamiltonian Mechanics.

\section*{Acknowledgments}

The author R. Azuaje wishes to thank CONAHCYT (Mexico) for the financial support through a postdoctoral fellowship in the program Estancias Posdoctorales por México 2022. A.M. Escobar Ruiz would like to thank the support from Consejo Nacional de Humanidades, Ciencias y Tecnologías (CONAHCyT) of Mexico under Grant CF-2023-I-1496 and from UAM research grant 2024-CPIR-0.

The authors thank Julio Gordiano for performing some of the computations in the examples.

\bibliography{refs} 

\providecommand{\noopsort}[1]{}\providecommand{\singleletter}[1]{#1}%
\begin{thebibliography}{10}

\bibitem{landau1982mechanics}
L.~D. Landau and E.~M. Lifshitz.
\newblock {\em Mechanics: Volume 1}.
\newblock Elsevier Science, 1982.

\bibitem{GPS2002}
H.~Goldstein, C.~Poole, and J.~Safko.
\newblock {\em Classical Mechanics, third edition}.
\newblock Addison-Wesley, 2002.

\bibitem{Calkin96}
M.~G. Calkin.
\newblock {\em Lagrangian and {H}amiltonian Mechanics}.
\newblock World Scientific Publishing, 1996.

\bibitem{Arnold78}
V.~I. Arnold.
\newblock {\em Mathematical Methods of Classical Mechanics}.
\newblock Springer New York, 1978.

\bibitem{AMRC2008}
R.~Abraham and J.~E. Marsden.
\newblock {\em Foundations of Mechanics}.
\newblock American Mathematical Soc. No. 364, 2008.

\bibitem{LR89}
M.~de~León and P.~R. Rodrigues.
\newblock {\em Methods of Differential Geometry in Analytical Mechanics}.
\newblock North-Holland Mathematics Studies, 1989.

\bibitem{Torres2020}
G.~F.~Torres del Castillo.
\newblock {\em Differentiable Manifolds: A Theoretical Physics Approach, second
  edition}.
\newblock Birkh\"{a}user Basel, 2020.

\bibitem{Lee2012}
J.~Lee.
\newblock {\em Introduction to Smooth Manifolds, second edition}.
\newblock Springer, 2012.

\bibitem{marsden2013introduction}
J.~E. Marsden and T.~S. Ratiu.
\newblock {\em Introduction to mechanics and symmetry: a basic exposition of
  classical mechanical systems}, volume~17.
\newblock Springer Science \& Business Media, 2013.

\bibitem{AKN2006}
V.~I. Arnold, V.~V. Kozlov, and A.~I. Neishtadt.
\newblock {\em Mathematical Aspects of Classical and Celestial Mechanics, third
  edition}.
\newblock Springer, 2006.

\bibitem{CFR2013}
J.~F. Cariñena, F.~Falceto, and M.~F. Rañada.
\newblock Canonoid transformations and master symmetries.
\newblock {\em J. Geo. Mech}, 5:151--166, 2013.

\bibitem{struckmeier2005hamiltonian}
J.~Struckmeier.
\newblock Hamiltonian dynamics on the symplectic extended phase space for
  autonomous and non-autonomous systems.
\newblock {\em J. Phys. A: Math. Gen.}, 38(6):1257, 2005.

\bibitem{AE2023}
R.~Azuaje and A.~M. Escobar-Ruiz.
\newblock Canonical and canonoid transformations for {H}amiltonian systems on
  (co)symplectic and (co)contact manifolds.
\newblock {\em J. Math. Phys.}, 64(3):033501, 2023.

\bibitem{azuaje2023scaling}
R.~Azuaje and A.~Bravetti.
\newblock Scaling symmetries and canonoid transformations in hamiltonian
  systems.
\newblock {\em Int. J. Geom. Methods Mod. Phys.}, 21(04):2450077, 2024.

\bibitem{CLL92}
F.~Cantrijn, M.~de~León, and E.~A. Lacomba.
\newblock Gradient vector fields on cosymplectic manifolds.
\newblock {\em J. Phys. A: Math. Gen.}, 25:175--188, 1992.

\bibitem{LS2017}
M.~de~León and C.~Sardon.
\newblock Cosymplectic and contact structures for time-dependent and
  dissipative {H}amiltonian systems.
\newblock {\em J. Phys. A: Math. Theor.}, 50:255205, 2017.

\bibitem{torres2018introduction}
G.~F.~Torres del Castillo.
\newblock {\em An Introduction to Hamiltonian Mechanics}.
\newblock Springer, 2018.

\bibitem{KS2011}
Y.~Kosmann-Schwarzbach and B.~E. Schwarzbach.
\newblock {\em The {N}oether Theorems, Invariance and Conservation Laws in the
  Twentieth Century}.
\newblock Springer New York, 2011.

\bibitem{BG2021}
A.~Bravetti and A.~Garcia-Chung.
\newblock A geometric approach to the generalized {N}oether theorem.
\newblock {\em J. Phys. A: Math. Theor.}, 54:095205, 2023.

\bibitem{winternitz1967symmetry}
P.~Winternitz, J.~A. Smorodisky, M.~Uhlir, and I~Fris.
\newblock On symmetry groups in classical and quantum mechanics. jadernaja fiz.
  4, 625--635 (1967).
\newblock {\em English transl.: Soviet J. Nucl. Phys}, 4:444--450, 1967.

\bibitem{bonatsos1994deformed}
D.~Bonatsos, C.~Daskaloyannis, and K.~Kokkotas.
\newblock Deformed oscillator algebras for two-dimensional quantum
  superintegrable systems.
\newblock {\em Phys. Rev. A}, 50(5):3700, 1994.

\end{thebibliography}
\bibliographystyle{unsrt} 

\end{document}